\documentclass[11pt]{amsart}
\usepackage{graphicx, amssymb, hyperref}

\newtheorem{thm}{Theorem}[section]

\newtheorem{prop}[thm]{Proposition}
\theoremstyle{definition}
\newtheorem{defn}[thm]{Definition}

\newtheorem{ass}[thm]{Assumption}

\theoremstyle{remark}
\newtheorem{rem}[thm]{Remark}

\numberwithin{equation}{section}

\newcommand{\abs}[1]{\left\vert#1\right\vert}
\newcommand{\set}[1]{\left\{#1\right\}}
\newcommand{\Real}{\mathbb R}
\newcommand{\Natural}{\mathbb N}

\newcommand{\A}{\mathcal{A}}
\newcommand{\B}{\mathcal{B}}
\newcommand{\such}{\, | \,}
\newcommand{\limn}{\lim_{n \to \infty}}

\newcommand{\nin}{{n \in \Natural}}
\newcommand{\tir}{{t \in \Real_+}}

\newcommand{\iii}{{i \in I}}
\newcommand{\jii}{{j \in I}}

\newcommand{\prob}{\mathbb{P}}

\newcommand{\hS}{\widehat{S}}

\newcommand{\Exp}{\mathcal E}
\newcommand{\qprob}{\mathbb{Q}}

\newcommand{\expec}{\mathbb{E}}
\newcommand{\expecp}{\expec_\prob}

\newcommand{\expecqi}{\expec_{\qprob^{i} } }
\newcommand{\expecqj}{\expec_{\qprob^{j} } }

\newcommand{\basis}{(\Omega, \, \mathbf{F})}
\newcommand{\basisp}{(\Omega, \, \mathbf{F}, \, \prob)}

\newcommand{\basisqj}{(\Omega, \, \mathbf{F}, \, \qprob^j)}

\newcommand{\F}{\mathcal{F}}

\newcommand{\cadlag}{c\`adl\`ag}

\newcommand{\ud}{\, \mathrm d}
\newcommand{\inner}[2]{\left \langle #1 , \, #2 \right \rangle}

\newcommand{\Stop}{\mathcal{T}}
\newcommand{\StopT}{\Stop_{[0, T]}}

\newcommand{\num}{num\'eraire}
\newcommand{\X}{\mathcal{X}}

\newcommand{\Y}{\mathcal{Y}}

\newcommand{\tX}{\widetilde{X}}

\newcommand{\tS}{\widetilde{S}}
\newcommand{\hX}{\widehat{X}}

\newcommand{\co}{\mathsf{c}}

\newcommand{\pare}[1]{\left(#1\right)}
\newcommand{\bra}[1]{\left[#1\right]}
\newcommand{\dbra}[1]{[\kern-0.15em[ #1 ]\kern-0.15em]}
\newcommand{\dbraco}[1]{[\kern-0.15em[ #1 [\kern-0.15em[}

\newcommand{\bF}{\mathbf{F}}

\newcommand{\indic}{\mathbf{1}}

\newcommand{\dfn}{ := }

\newcommand{\absco}{{<\kern-0.53em<}}

\newcommand{\naone}{NA$_1$}

\newcommand{\indicz}{\indic_{\dbraco{0, \zeta}}}

\newcommand{\EX}{\mathsf{EX}}
\newcommand{\AX}{\mathsf{AX}}

\newcommand{\zi}{(0, \infty]}
\newcommand{\Opt}{\mathcal{O}}

\newcommand{\on}{{; \,}}

\newcommand{\Sym}{\mathbb{S}}

\newcommand{\I}{\mathcal{I}}

\newcommand{\oPhi}{\overline{\Phi}}
\newcommand{\oney}{\mathbf{1}_{\{ |y| \leq 1 \}}}

\begin{document}

\title[]{Valuation and parities for exchange options}%
\author{Constantinos Kardaras}%
\address{Constantinos Kardaras, Department of Statistics, London School of Economics and Political Science, 10 Houghton st, London, WC2A 2AE, UK.}%
\email{k.kardaras@lse.ac.uk}%

\date{\today}%
 \thanks{I am grateful to Johannes Ruf and two anonymous referees for valuable discussions that improved the presentation of the article.}%

\date{\today}%
\begin{abstract}
Valuation and parities for both European-style and American-style exchange options are presented in a general financial model allowing for jumps, possibility of default and ``bubbles'' in asset prices. The formulas are given via expectations of auxiliary probabilities using the change-of-\num \ technique. Extensive discussion is provided regarding the way that folklore results such as Merton's no-early-exercise theorem and traditional parities have to be altered in this more versatile framework.
\end{abstract}

\maketitle


\section*{Introduction}

A multitude of contracts in financial markets can be regarded as options to exchange units of one asset for certain units of another. The first paper to discuss and consider such options is \cite{RePEc:bla:jfinan:v:33:y:1978:i:1:p:177-86}. Building upon the ground-breaking methodology of \cite{citeulike:202505} and \cite{MR0496534}, formulas were provided for the fair value of exchange options for two no-dividend-paying assets in a Black-Scholes-Merton modelling environment. Depending on which of the two assets is chosen as a \num \ in order to denominate wealth, such exchange options can be regarded either of a call or a put type. Under this perspective, and always in the Black-Scholes-Merton model, Merton's no-early-exercise result \cite[Theorem 2]{MR0496534} can be seen to imply that American-style exchange options have the same value as their European-style counterparts; then, the usual put-call parity translates to a single parity between exchange options of either European or American style.

In recent literature, considerable interest has been placed in financial models where certain anomalies exist, a prominent one concerning assets which contain bubbles---see, for example, \cite{MR1339738}, \cite{MR2213778}, \cite{MR2653260}, \cite{MR2732840}, \cite{MR2359365}, \cite{MR2650245}, \cite{KKN}. 
An inspection of papers on the subject reveals several possible directions that one may proceed in the mathematical definition of a bubble. While no attempt will be made here to summarize or consolidate these views, only for illustration purposes we mention the specialized case of complete markets (see, for example, \cite{MR2359365}), where the different definitions essentially coincide: a certain asset contains a bubble if the market allows for arbitrage relative to its cum-dividend\footnote{For notational simplicity, in this paper only no-dividend paying assets are treated.} price process; in other words, there exist free snacks (in the terminology of \cite{MR1774056}) relative to the asset with the bubble. This last fact prevents the existence of an \emph{equivalent} probability which would render the (local) martingale property to wealth processes denominated in units of the asset containing the bubble. Such probability measures are used for valuation of illiquid financial derivative securities; therefore, it would appear that existence of baseline assets containing bubbles presents a hurdle in the development of the theory of financial mathematics. However, a consistent theory of valuation and hedging can still be developed in models where assets with bubbles exist, provided that one utilizes strictly positive local martingale deflators instead of equivalent local martingale measures---the survey article \cite{KarFern} is a thorough reference in this respect.\footnote{Even when an equivalent local martingale measure exists in the market, it may still be the case that some strictly positive local martingale deflator which is not an actual martingale is used for valuation. Indeed, this may happen in cases where utility indifference valuation rules are considered, as is explained in \cite{MR2132189}.} Under appropriate assumptions on the underlying stochastic environment which allow for the inference of existence of probability measures in the spirit of Kolmogorov's extension theorem (as explained, for example, in \cite{para}) local martingale deflators can still define auxiliary probabilities that can be used for valuation; see \cite{exit_m} and Theorem \ref{thm: prob_existnz} later on. It should be noted, however, that these valuation probabilities may fail to be even locally (along a sequence of deterministic times converging to infinity) equivalent to the original probability.

Several results that are folklore in traditional models fail to hold when valuation is done using strictly positive local martingale deflators as opposed to local martingale measures; typical examples of such failure include the aforementioned no-early-exercise theorem for American options, as well as certain parities---for a discussion, see \cite{MR2903625}. However, it is becoming increasingly understood that an alternative viewpoint concerning such results enables the provision of formulas that are valid in these wider-encompassing models, allowing for valuation using local martingale deflators. Such viewpoint also facilitates the understanding of the exact attributes of earlier models that resulted in such formulas. The present paper contributes to the existing literature by providing valuation and parities for exchange options via the change-of-\num \ approach in a general modelling environment where equivalent martingale measures may fail to exist, allowing for jumps and possible default. As mentioned previously, in order to provide formulas in terms of expectations under auxiliary valuation probabilities, mild assumptions have to be enforced on the underlying filtered measurable space---canonical examples of such environments are Markovian models driven by economic factors, a case that is discussed in detail in the paper. Due to the potential failure of existence of equivalent martingale measures with respect to some assets, the value of American exchange options may be higher than the corresponding value of exchange options of European type; a general formula for the early exercise premium (in terms of explosion probabilities, amongst other elements) is provided that covers all models. The latter discrepancy of American and European option values affects the parities: several different parities relating European and American exchange option values are provided.

\smallskip

The structure of the paper is as follows. Section \ref{sec: main} presents the underlying financial framework, while Section \ref{sec: probs} establishes existence of the valuation probabilities and studies the behaviour of ratios of asset prices under these probabilities. In Section \ref{sec: valuation}, several formulas for valuation of European and American exchange options are presented. Finally, Section \ref{sec: parity} explores the different parities between exchange options of both European and American type, including an example involving the three-dimensional Bessel process where explicit formulas are available.

\section{Underlying Framework} \label{sec: main}

\subsection{The set-up} \label{subsec: prelims}

In the later development of the paper, the need will arise to infer existence of probabilities arising from local martingale density processes; in order to ensure such existence, we shall require a special structure for the underlying probability space, which we introduce below.

The set of all possible states of the financial environment is modelled through a Polish space $E$. Consider an additional isolated point $\triangle$ that is appended to $E$ and will model a ``cemetery'' state for the economy. If $\omega: [0, \infty) \mapsto E \cup \set{\triangle}$ is a right-continuous function, define
\[
\zeta (\omega) \dfn \inf \set{\tir \such \omega(t) = \triangle},
\]
where $\zeta$ has the interpretation of the economy's lifetime. With this understanding, let $\Omega$ denote the set of all right-continuous functions $\omega: [0, \infty) \mapsto E \cup \set{\triangle}$ such that $\omega(0) \in E$ and $\omega (t) = \triangle$ holds for all $t \in [\zeta(\omega), \infty)$; in words, $\Omega$ consists of right-continuous paths which are at $E$ at time zero, and remain forever in the cemetery state $\triangle$, once reached. Note that $\zeta(\omega) \in \zi$ holds for all $\omega \in \Omega$.

We denote by $Z = (Z_t)_{\tir}$ the co\"ordinate process on $\Omega$, i.e., for fixed $\tir$ it holds that $Z_t(\omega) = \omega(t)$, for all $\omega \in \Omega$. 
Define $\bF = (\F_t)_{\tir}$ as the right-continuous augmentation of the smallest filtration that makes $Z$ adapted.
Define also $\F \dfn \bigvee_{\tir} \F_t$. We denote by $\Stop$ will the class of all (possibly infinite-valued) stopping times on $\basis$. Note that the $\zeta \in \Stop$. 

\begin{rem} \label{rem: cont_path_factor}
If a model having factors that change in a continuous fashion is desired, $\Omega$ can be chosen to consist of right-continuous functions that are actually continuous on $[0, \zeta[$.
\end{rem}

The following interpretation should be kept in mind throughout the paper: from time $\zeta$ onwards, all economic activity ceases and no financial claims are honoured. Incorporating (stochastic) recovery rate at default is also possible within the present framework; however, we decide to only treat the case of no recovery in order to allow for some simplification in the presented formulas.

\subsection{Assets and stochastic discount factor}

\label{subsec: assets}

On the filtered measurable space $\basis$ satisfying the tenets of Subsection \ref{subsec: prelims}, we postulate the existence of nonnegative \cadlag \ processes $S^i$ for $\iii$, where $I$ is an arbitrary finite index set. Each $S^i$, $\iii$, is modelling the price-process of a no-dividend-paying 
asset in the financial market. All assets are denominated in the same \num, which will not actually play any role in our treatment since from Section \ref{sec: probs} onwards the assets $(S^i)_{\iii}$ themselves are going to be used as \num s. (See also the discussion after Assumption \ref{ass: main}.)

To keep in par with the interpretation of $\zeta$ as the economy's lifetime, it shall be assumed that $S^i = 0$ holds on the stochastic interval $\dbraco{\zeta, \infty} \  \dfn \set{(\omega, t) \in \Omega \times \Real_+ \such \zeta(\omega) \leq t}$ for all $\iii$.\footnote{This fact is repeated in  Assumption \ref{ass: main} below.} Earlier default for a specific asset is also possible in our framework.

The full probabilistic model for the movement of the asset prices is described by the introduction of a probability $\prob$ on the $\sigma$-algebra $\F$. The symbol ``$\expecp$'' denotes expectation with respect to $\prob$, with analogous notation used for expectation under other probabilities that will eventually appear. Expressions of the form $\expec \bra{\xi \on A}$ for nonnegative $\F$-measurable random variable $\xi$ and $A \in \F$ are shorthand notations for $\expec \bra{\xi \indic_A}$, where ``$\indic_A$'' denotes the indicator of $A$.

The following will be a standing assumption throughout the paper.

\begin{ass} \label{ass: main}
For all $\iii$, $S^i = 0$ holds on $\dbraco{\zeta, \infty}$ and $S^i_0$ is $\prob$-a.s. constant and strictly positive. Furthermore, there exists a nonnegative process $Y$ with $\prob \bra{Y_0 = 1} = 1$ such that $Y S^i$ is a $\prob$-a.s. (\cadlag) local martingale on $\basisp$ for all $\iii$.
\end{ass}

For the remainder of \S \ref{subsec: assets}, we discuss the economic significance of Assumption \ref{ass: main}. In the process, and in order to use previous classic results, additional structural assumptions shall be made. We stress, however, that none of the extra assumptions that appear below will be needed in the remainder of the text.

It is traditional in the field of Mathematical Finance to choose one of the assets $(S^i)_{\iii}$ as a \num, in order to denominate all other wealth; this \num \ is supposed to stay strictly positive for all times that the economy is alive. Wanting to keep symmetry in our framework and be par with classical theory, we define $S^* \dfn \pare{\sum_{\iii} S^i} / \pare{\sum_{\iii} S_0^i}$, and make the additional assumptions that $\set{S^* > 0} = \dbraco{0, \zeta}$ holds up to a $\prob$-indistinguishable set and $\prob \bra{\zeta < \infty} = 0$.\footnote{Even when $\prob \bra{\zeta < \infty} = 0$ holds, the introduction of the cemetery state $\triangle$ in our framework is essential, since the event $\set{\zeta < \infty}$ may have non-zero measure under the probabilities $(\qprob^i)_{\iii}$ that are constructed in Theorem \ref{thm: prob_existnz}.} In this case, define $\hS \dfn (\hS^i)_{\iii}$ via $\hS^i \dfn (S^i / S^*) \indicz$ for all $\iii$, and assume that $\hS$ is a $d$-dimensional semimartingale. For a predictable and $\hS$-integrable process $H$, set
\[
\hX^{H} \dfn \pare{ 1 + \int_{(0, \cdot]} \sum_{\iii} H^i_t \ud \hS^i_t } \indicz,
\]
with the understanding that vector stochastic integration is used. The previous expression for $\hX^H$ gives the value of the portfolio generated by the strategy $H$, denominated in terms of $S^*$. We also define $X^{H} \dfn S^* \hX^{H}$, as well as $\X$ to be the class of all \emph{nonnegative} $X^{H}$, where $H$ is any predictable $\hS$-integrable process $H$. The class $\X$ contains all nonnegative wealth processes that are denominated in the same units as all the assets with price-processes $(S^i)_{\iii}$.\footnote{It is important to note that the class $\X$ does \emph{not} depend on the specific choice of \num, as we now explain. Suppose that  $X^* \in \X$ is such that $\{ X^* > 0 \} = \dbraco{0, \zeta}$ holds up to a $\prob$-indistinguishable set. Define $\tS \dfn (\tS^i)_{\iii}$ via $\tS^i \dfn (S^i / X^*) \indicz$, and note that $\tS$ is a $d$-dimensional semimartingale, since $\hS$ is. For a predictable and $\tS$-integrable process $H$, set $\tX^{H} \dfn \pare{ 1 + \int_{(0, \cdot]} \sum_{\iii} H^i_t \ud \tS^i_t } \indicz$, which gives the value of the portfolio denominated in terms of $X^*$. It can be checked that $\X$ coincides with the class of all nonnegative $X^* \tX^{H}$, where $H$ ranges through all predictable $\tS$-integrable processes.}

Along with $\X$, define the class of all local martingale deflators $\Y$ via
\[
\Y \dfn \set{Y \geq 0 \such Y_0 = 1 \text{ and } Y X \text{ is a \cadlag \ $\bF$-local $\prob$-martingale for all } X \in \X}.
\]
Under very mild condition on the absence of ``free lunches'' in the market, one can ensure that $\Y$ is non-empty and, in fact, contains a strictly positive process.\footnote{Condition NFLVR of \cite{MR1304434} will certainly be sufficient. More precisely, the weaker (than NFLVR) condition \naone \ is equivalent to the statement that $\Y$ contains a strictly positive process; for example, see \cite{MR2972237} or \cite{MR3177411}.} The set $\Y$ is of importance in the problem of utility maximization and elements of $\Y$ can be used in order to compute utility indifference prices---the interested reader should check \cite{MR1722287} and \cite{MR2132189} for these facts. Processes in $\Y$, as in Assumption \ref{ass: main}, are commonly referred to as \textsl{stochastic discount factors}, and are used for the valuation of financial derivatives.

%
%

\subsection{Markovian factor models} \label{exa: Markov_factor}

We discuss here the validity of Assumption \ref{ass: main} in a wide range of continuous-time Markovian factor models with possible jumps and default. We shall specialize the framework of Subsection \ref{subsec: prelims} to the case $E = \Real^m$ for some $m \in \Natural$. Recall that $Z = (Z_t)_{\tir}$ denotes the coordinate process on $\Omega$. 

Consider a bounded measurable $a:\Real^m \mapsto \Real^m$, a bounded continuous $c:\Real^m \mapsto \Sym^m_{++}$, where $\Sym^m_{++}$ denotes the space of strictly positive definite symmetric $m \times m$ matrices, as well as $\nu: \Real^m \times \B(\Real^m) \mapsto \Real_+$ such that $\nu(z, \cdot)$ is a $\sigma$-additive measure on $\B(\Real^m)$ for all $z \in \Real^m$, and 
\[
\Real^m \ni z \mapsto \int_\Gamma \pare{1 \wedge |y|^2} \nu(z, \ud y) \text{ is continuous and bounded}, \quad \forall \Gamma \in \B(\Real^m).
\]
With the above notation, and with $\inner{\cdot}{\cdot}$ denoting (sometimes, formally) inner product on $\Real^m$, define the operator $C^\infty_0 (\Real^m) \ni f \mapsto \A(f)$ such that
\begin{align} \label{eq:infin_gen}
\A(f)(z) &\dfn \inner{a}{\nabla f}(z) + \frac{1}{2} \sum_{k=1}^m \sum_{l=1}^m c^{k l} (z) \frac{\partial^2 f}{\partial z^{kl}} (z) \\
\nonumber &+ \int_{\Real^d} \pare{f(z + y) - f(z) - \inner{ y}{\nabla f(z)} \oney } \nu(z, \ud y), \quad f \in C^\infty(\Real^m), \ z \in \Real^m.
\end{align}
Finally, fix a measurable and locally bounded function $\lambda: \Real^m \mapsto \Real_+$ and $z_0 \in \Real^m$. With the previous notation and assumptions, and with $\I$ denoting the identity operator on $C^\infty_0 (\Real^m)$, there exists a unique solution $\prob$ to the martingale problem  associated with $\A - \lambda \I$, with ``killing'' rate function $\lambda$, such that $\prob \bra{Z_0 = z_0} = 1$.\footnote{When $\lambda \equiv 0$, one can consult \cite[Theoreme (13.58)]{MR542115} or \cite[Thereom 4.3]{MR0433614} for existence of a unique solution $\prob^0$ to the corresponding martingale problem. Once the probability $\prob^0$ is constructed, for which $\prob^0 \bra{\zeta < \infty} = 0$, one may extend the probability space and introduce an independent (of $Z$) random variable $\eta$ with unit-rate exponential law, then set
\[
\xi \dfn \inf \set{t \in \Real_+ \ \Big| \ \int_0^t \lambda(Z_{s-}) \ud s > \eta},
\]
and then define $\widetilde{Z} = Z \indic_{\dbraco{0, \xi}} + \triangle \indic_{\dbraco{\xi, \infty}}$. Finally, one defines $\prob$ to be the law of $\widetilde{Z}$ under $\prob^0$ on the canonical space $(\Omega, \F)$, and note that indeed $\prob$ solves the martingale problem associated with $\A - \lambda \I$, satisfying $\prob \bra{Z_0 = z_0} = 1$.} In particular, the continuous part of the quadratic covariation process of $Z$ is given by $\int_0^{\zeta \wedge \cdot} c(Z_{t-}) \ud t$ and the compensator of the jump measure $\mu$ of $Z$ is equal to $\int_0^{\zeta \wedge \cdot} \nu (Z_{t-}, \ud y) \ud t$.

The $m$-dimensional factor process $Z$ will drive the prices of $(d+1)$ financial assets, where $d \in \Natural$. Let $I = \set{0, 1, \ldots, d}$; the index ``$0$'' is reserved for a locally riskless asset, which is typical in the literature. 
Consider a \emph{short rate} function $r: \Real^m \mapsto \Real$; furthermore, for $\iii$ consider \emph{excess rate of return} functions $\alpha^i : \Real^m \mapsto \Real$, functions $\beta^i : \Real^m \mapsto \Real^{m}$ that will control the continuous part of the quadratic variation of the assets and functions $\gamma^i : \Real^m \times \Real^m \mapsto \Real_+$ that will control the relative jump sizes of the asset prices. For reasons of unifying presentation, set also $\alpha^0 : \Real^m \mapsto \Real$, $\beta^0 : \Real^m \mapsto \Real^m$ and $\gamma^0 : \Real^m \times \Real^m \mapsto \Real_+$ to be identically equal to zero. The previous functions are assumed measurable and such that
\begin{equation} \label{eq:assumpt_stock_coeff}
\sup_{z \in (-n,n)^m} \pare{\abs{\pare{r + \alpha_i}(z)} + \inner{\beta^i}{c \beta^i} (z) + \int_{\Real^m} \pare{\gamma^i(z, y) - 1}^2 \nu(z, \ud y)} < \infty, \quad \forall \nin, \ \iii.
\end{equation}
Define processes $(S^i)_{\iii}$, satisfying $S^i = S^i_0 \Exp(U^i) \indicz$ for all $\iii$, where $S^i_0 > 0$, ``$\Exp$'' denotes the stochastic exponential operator throughout, and
\begin{equation} \label{eq: Markov_asset_dyn}
U^i =  \int_0^\cdot \pare{\pare{r+ \alpha^i} (Z_{t-})  \ud t + \inner{\beta^i (Z_{t-})}{\ud Z^\co_t} + \int_{\Real^m} \pare{\gamma^i(Z_{t-}, y) - 1} \pare{\mu(\ud y, \ud t) - \nu(Z_{t-}, \ud y) \ud t}}
\end{equation}
holds for all $\iii$, where the previous process is well defined in view of \eqref{eq:assumpt_stock_coeff}. As was already mentioned, the functions $(\beta^i)_{\iii}$ control the continuous part of the local covariation between the asset-price movement and the driving economic factors, as well as the other assets. Accordingly, since $S^i = S^i_{-} \gamma^i(Z_-, \Delta Z)$ holds for $\iii$, the functions $(\gamma^i)_{\iii}$ control the jumps in the asset-price movement given abrupt changes in the underlying economic factors. Note that individual assets may default before time $\zeta$, since we allow for the opportunity that $\gamma^i$ takes the value $0$ and $\nu(\cdot, \ud y)$ may have atomic parts.


In order to define the stochastic discount factor, consider measurable functions $\phi: \Real^m \mapsto \Real^m$ and $\psi: \Real^m \times \Real^m \mapsto (0, \infty)$ with
\begin{equation} \label{eq:assumpt_disc_coeff}
\sup_{z \in (-n,n)^m} \pare{\inner{\phi}{c \phi} (z) + \int_{\Real^m} \pare{\psi(z, y) - 1}^2 \nu(z, \ud y)} < \infty, \quad \forall \nin.
\end{equation}
as well as\footnote{The existence of at least one pair of functions $(\phi, \psi)$ that satisfy \eqref{eq:market_price_risk} follows directly from no-arbitrage considerations. We do ask that one can choose such a pair satisfying the extra local boundedness conditions \eqref{eq:assumpt_disc_coeff}, which is a rather mild technical assumption. 
}
\begin{equation} \label{eq:market_price_risk}
\alpha^i(z) + \inner{\phi(z)}{c (z)\beta^i (z)} + \int_{\Real^m}  \pare{\gamma^i (z, y) -1} \pare{\psi(z, y) -1} \nu (z, \ud y) = 0, \quad \forall z \in \Real^m \text{ and } \iii.
\end{equation}
Define the process $Y$ satisfying $Y = \Exp(V) \indicz$, where
\begin{equation} \label{eq: st_disc_fac_markov}
V = \int_0^\cdot \pare{(\lambda - r )(Z_{t-}) \ud t + \inner{\phi (Z_{t-})}{\ud Z^\co_t} + \int_{\Real^m} \pare{\psi(Z_{t-}, y) -1} \pare{\mu(\ud y, \ud t) - \nu(Z_{t-}, \ud y) \ud t}},
\end{equation}
where the previous process is well defined in view of \eqref{eq:assumpt_disc_coeff} and the fact that $\lambda$ is locally bounded. A straightforward use of the integration-by-parts formula shows that $Y S^i = S^i_0 \Exp(V^i) \indicz$, where
\begin{align} \label{eq:densities_markov}
V^i &= \int_0^\cdot \lambda(Z_{t-}) \ud t  + \int_0^\cdot \inner{(\beta^i + \phi) (Z_{t-})}{\ud Z^\co_t} \\ 
\nonumber &+ \int_0^\cdot \pare{ \int_{\Real^m} \pare{\gamma^i (Z_{t-}, y) \psi  (Z_{t-}, y) - 1} \pare{\mu(\ud y, \ud t) - \nu(Z_{t-}, \ud y) \ud t}}, \quad \forall \iii.
\end{align}
Define the nondecreasing sequence $(\zeta_n)_{\nin}$ of stopping times via
\begin{equation} \label{eq: zeta_n_explicit}
\zeta_n \dfn \inf \set{\tir \such Z_t \notin (-n, n)^m}, \quad \text{for } \nin.
\end{equation}
It is straightforward to check that $\zeta_n \leq \zeta$, $\prob \bra{\limn \zeta_n = \zeta} = 1$ and $\limn \prob \bra{\zeta_n < \zeta < \infty} = 0$. Furthermore, \eqref{eq:assumpt_stock_coeff}, \eqref{eq:assumpt_disc_coeff}, the fact that $\gamma^i \psi - 1 = (\gamma^i - 1)(\psi - 1) + (\gamma^i - 1) + (\psi - 1)$  holds identically (which incidentally also shows that the last stochastic integral in \eqref{eq:densities_markov} is well-defined), and local boundedness of $\lambda$ implies by \cite[Theorem 12]{MR515738} that the processes $\big( Y_{\zeta_n \wedge \cdot} S_{\zeta_n \wedge \cdot}^i \big)_{\tir}$ are (true) martingales for all $\nin$ and $\iii$. In particular, $Y S^i$ is a local martingale on $\basisp$ for all $\iii$, and one obtains the validity of Assumption \ref{ass: main} in this extremely versatile setting.



\section{Valuation Probabilities and Asset Ratios} \label{sec: probs}

\subsection{Valuation probabilities} 

As mentioned in Subsection \ref{subsec: assets}, the process $Y$ of Assumption \ref{ass: main} plays the role of a stochastic discount factor in the market. As such, it will be used for valuation of securities: the present (time zero) value a contract that pays an $\F_T$-measurable nonnegative amount $H_T$ at time $T \in \Stop$ is $\expecp \bra{Y_T H_T \on T < \zeta}$.  It is customary to write valuation formulas in terms of expectation under auxiliary valuation probabilities. In order to obtain the latter from the representation in terms of expectations under $\prob$ and stochastic discounting, a ``baseline'' (or ``\num'') asset has to be chosen in order to denominate wealth. Section \ref{sec: valuation} and Section \ref{sec: parity} deal with valuation and parities for exchange options; for this reason, we refrain from choosing a single asset to use as baseline; rather, a family of probabilities $(\qprob^i)_{\iii}$ will be introduced, one for each asset indexed by $\iii$ being used as a baseline. Care has to be exercised in defining these probabilities, since the candidate ``density processes'' that have to be used in defining them are in general only \emph{local} martingales on $\basisp$. However, as stated in Theorem \ref{thm: prob_existnz} below, the structure of the filtered probability space described in Subsection \ref{subsec: prelims} allows for such construction under Assumption \ref{ass: main}. A proof of Theorem \ref{thm: prob_existnz} in this exact setting appears in \cite{BBKN}; of course, results of similar nature have appeared previously---see, for example, \cite{exit_m}, \cite{Meyer}, \cite{MR1339738}, \cite{MR2653260}, and \cite{Ruf-Perk} 

Before the statement of Theorem \ref{thm: prob_existnz}, recall that 
the optional sigma-field $\Opt$ on $\Omega \times \Real_+$ is the one generated by all \cadlag \ processes; then, a process is called optional if it is $\Opt$-measurable.

\begin{thm} \label{thm: prob_existnz}
Under Assumption \ref{ass: main}, for each $\iii$ there exists a unique probability $\qprob^i$ on $(\Omega, \F)$ such that the following property is valid: for any nonnegative optional process $H$ on $\basis$,
\begin{equation} \label{eq: measure_change}
\expecp \bra{Y_T H_T S^i_T \on T < \zeta} = S^i_0  \expecqi \bra{ H_T \on T < \zeta} \quad \text{holds for all } T \in \Stop.
\end{equation}
\end{thm}

\begin{rem} \label{rem: zero_prob_zero}
For $\iii$ and $T \in \Stop$, a use of \eqref{eq: measure_change} with $H = \indic_{\set{S^i = 0}}$ gives $\qprob^i \bra{S^i_T = 0, \, T < \zeta} = 0$.
\end{rem}

\begin{rem} \label{rem: no_expl_iff_mart}
Assumption \ref{ass: main} and a straightforward application of the conditional version of Fatou's lemma implies that $Y S^i$ is a (nonnegative) supermartingale on $\basisp$ for all $\iii$. Using $H \equiv 1$ in \eqref{eq: measure_change} and taking $T \in \Stop$ to be equal to $\tir$, it follows that $S_0^i \qprob^i \bra{t < \zeta } = \expecp \bra{Y_t S^i_t \on t < \zeta } = \expecp \bra{Y_t S^i_t}$ holds for all $\tir$ and $\iii$, where the last equation follows from the fact that $S^i_t = 0$ holds on $\set{\zeta \leq t}$. It then follows in a straightforward way that $\qprob^i \bra{\zeta < \infty} = 0$ holds for some $\iii$ if and only if the process $(Y_t S^i_t)_{\tir}$ is a (true) martingale on $\basisp$.
\end{rem}

\subsection{Foretellability of $\zeta$} \label{subsec:foretellability}

In general, $\zeta$ is \emph{not} a predictable\footnote{Following standard terminology from the general theory of stochastic processes (see, for example,  \cite{MR1943877}), a stopping time $\tau \in \Stop$ is predictable on $\basis$ if the stochastic interval $\dbraco{\tau, \infty}$ is a predictable set; note that this notion does not take into account any underlying probability on the filtered probability space.} stopping time on $\basis$. However, as we shall see, it actually is predictable on $(\Omega, \, \bF^{\qprob^i})$ for all $\iii$, where $\bF^{\qprob^i}$ is the usual $\qprob^i$-augmentation of $\bF$. We first give an essential definition.

\begin{defn}
Let $\qprob$ be a probability on $(\Omega, \F)$. A nondecreasing sequence $(\zeta_n)_{\nin}$ of stopping times will be said to \textsl{foretell $\zeta$ under $\qprob$} if $\zeta_n \leq \zeta$ for all $\nin$, $\qprob \bra{\zeta_n < \zeta, \, \forall \nin} = 1$ and $\qprob \bra{\limn \zeta_n = \zeta} = 1$.
\end{defn}

\begin{prop} \label{prop: foretellable}
Let $(\zeta_n)_{\nin}$ be a nondecreasing sequence of stopping times such that $\zeta_n \leq \zeta$ for all $\nin$, $\prob \bra{\limn \zeta_n = \zeta} = 1$ and $\big( Y_{\zeta_n \wedge t} S^i_{\zeta_n \wedge t} \big)_{\tir}$ is a uniformly integrable martingale on $\basisp$ for all $\iii$. Then, $(\zeta_n)_{\nin}$ foretells $\zeta$ under each of the probabilities $\qprob^i$, $\iii$.
\end{prop}

\begin{proof}
Applying \eqref{eq: measure_change} with $H \equiv 1$ and $T = \zeta_n$ gives $S_0^i \qprob^i \bra{ \zeta_n < \zeta} = \expecp \big[ Y_{\zeta_n} S^i_{\zeta_n} ; \zeta_n < \zeta \big]$, for all $\nin$ and $\iii$. Since $S^i_{\zeta_n} \indic_{\set{\zeta_n = \zeta}} = S^i_{\zeta} \indic_{\set{\zeta_n = \zeta}} = 0$, $\qprob^i \bra{ \zeta_n < \zeta} = (1 / S^i_0) \expecp \big[ Y_{\zeta_n} S^i_{\zeta_n} \big] = 1$ follows for all $\nin$ and $\iii$. Therefore, $\qprob^i \bra{\zeta_n < \zeta, \, \forall \nin} = 1$. Continuing, let $\zeta_\infty \dfn \limn \zeta_n$. Another application of \eqref{eq: measure_change}  with $H \equiv 1$ and $T = \zeta_\infty$ gives $S_0^i \qprob^i \bra{ \zeta_\infty < \zeta} = \expecp \big[ Y_{\zeta_\infty} S^i_{\zeta_\infty} ; \zeta_\infty < \zeta \big] = 0$ for all $\iii$, in view of the fact that $\prob \bra{\zeta_\infty < \zeta} = 0$. Therefore, we obtain $\qprob^i \bra{\zeta_\infty < \zeta} = 0$, or $\qprob^i \bra{\zeta_\infty = \zeta} = 1$ for all $\iii$, which shows that $(\zeta_n)_{\nin}$ of stopping times that foretells $\zeta$ under each of the probabilities $\qprob^i$, $\iii$.
\end{proof}

\begin{rem} \label{rem:foretellable_exists}
Note that sequences $(\zeta_n)_{\nin}$ satisfying the tenets of Proposition \ref{prop: foretellable} certainly exist. For a specific example, define
\[
\zeta_n \dfn \inf \set{\tir \ \big| \  Y_t \max_{\iii} S^i_t > n} \wedge \zeta, \quad \forall \nin.
\]
\end{rem}

\begin{rem} \label{rem:localisation_prob}
Let $(\zeta_n)_{\nin}$ be any localising sequence as in Proposition \ref{prop: foretellable}. With $L^i \dfn Y S^i / S^i_0$ for $\iii$, \eqref{eq: measure_change} implies that $L^i_{\zeta_n} = L^i_{\zeta_n} \indic_{\set{\zeta_n < \zeta}}$ holds for all $\iii$ and $\nin$. Therefore, $L^i_{\zeta_n}$ is the density of $\qprob^i$ with respect to $\prob$ on $\F_{\zeta_n}$ for all $\iii$ and $\nin$. This fact can help in obtaining the behaviour of processes under $\qprob^i$ for $\iii$; see Example \ref{exa: Markov_factor_reprise} below for an illustration.

Although $\qprob^i$ is absolutely continuous with respect to $\prob$ on $\F_{\zeta_n}$ for all $\iii$ and $\nin$, it should be noted that there is no general relationship between $\qprob^i$ and $\prob$ on $\F$.
\end{rem}

\subsection{Markovian factor models, continued} \label{exa: Markov_factor_reprise}

We proceed with an illustration of Theorem \ref{thm: prob_existnz} in the framework of Subsection \ref{exa: Markov_factor}, from which we retain all notation. Let $Y = \Exp(V) \indicz$, where $V$ is given in \eqref{eq: st_disc_fac_markov}, and fix $\jii$. Recall the sequence $(\zeta_n)_{\nin}$ of \eqref{eq: zeta_n_explicit}. By Proposition \ref{prop: foretellable}, $(\zeta_n)_{\nin}$ foretells $\zeta$ under $\qprob^j$.

A straightforward use of Girsanov's theorem via localization (see Remark \ref{rem:localisation_prob}) over $(\zeta_n)_{\nin}$ implies that $Z$ under $\qprob^j$ solves the martingale problem with possible explosion associated with the operator $\A^{\qprob^j}$ that is given as in \eqref{eq:infin_gen}, with $a$ there replaced by $a_{\qprob^j} \dfn a + c (\beta^j + \phi) + \int_{\Real^m} \pare{\gamma^j (\cdot, y) \psi  (\cdot, y) - 1} y \oney \nu(\cdot, \ud y)$, $c$ staying the same, and $\nu(\cdot, \ud y)$ replaced by $\nu_{\qprob^j} (\cdot, \ud y) \dfn \gamma^j (\cdot, y) \psi  (\cdot, y) \nu(\cdot, \ud y)$. Furthermore, recalling \eqref{eq: Markov_asset_dyn}, and noting (again, as a consequence of Girsanov's theorem) that the continuous local martingale part of $Z$ under $\qprob^j$ on $\dbraco{0, \zeta}$ is\footnote{Note that the process $Z^{\co, \qprob^j}$ is well defined and finitely-valued on the stochastic interval $\dbraco{0, \zeta}$, which is indistinguishable from $\bigcup_{\nin} \dbra{0, \zeta_n}$ under $\qprob^j$; however, it may happen that it explodes at $\zeta$.} $Z^{\co, \qprob^j} = Z_{\zeta \wedge \cdot}^\co - \int_0^{\zeta \wedge \cdot} c (Z_{t-}) (\beta^i + \phi)(Z_{t-}) \ud t$, we obtain that $S^i = S^i_0 \Exp(U^i) \indicz$, where $U^i$ is defined on $\dbraco{0, \zeta}$ via
\begin{align*}
U^i &=  \int_0^\cdot  \pare{r+ \alpha^i + \inner{\beta^i}{c (\beta^j + \phi)}} (Z_{t-}) \ud t \\
&+ \int_0^\cdot \pare{ \int_{\Real^m} \pare{\gamma^i(Z_{t-}, y) - 1} \pare{\gamma^j(Z_{t-}, y) \psi(Z_{t-}, y) - 1} \nu(Z_{t-}, \ud y)}   \ud t \\
&+ \int_0^\cdot \pare{ \inner{\beta^i (Z_{t-})}{\ud Z^{\co, \qprob^j}_t} + \int_{\Real^m} \pare{\gamma^i(Z_{t-}, y) - 1} \pare{\mu(\ud y, \ud t) - \nu_{\qprob^j} (Z_{t-}, \ud y) \ud t}}
, \quad \forall \iii.
\end{align*}
Recalling \eqref{eq:market_price_risk}, after simple algebra we obtain that, on $\dbraco{0, \zeta}$,
\begin{align*}
U^i &=  \int_0^\cdot \pare{ \pare{r + \inner{\beta^i}{c \beta^j}} (Z_{t-} ) +  \int_{\Real^m} \pare{\gamma^i(Z_{t-}, y) - 1} \pare{\gamma^j(Z_{t-}, y) - 1} \psi(Z_{t-}, y)  \nu(Z_{t-}, \ud y)} \ud t \\
&+ \int_0^\cdot \pare{ \inner{\beta^i (Z_{t-})}{\ud Z^{\co, \qprob^j}_t} + \int_{\Real^m} \pare{\gamma^i(Z_{t-}, y) - 1} \pare{\mu(\ud y, \ud t) - \nu_{\qprob^j} (Z_{t-}, \ud y) \ud t}}
, \quad \forall \iii.
\end{align*}
In the setting of this example, note that $S^i / S^j = (S^i_0 / S^j_0) \Exp(U^{ij})$ holds on $\dbraco{0, \zeta}$, where
\begin{align*}
U^{ij} &=   \int_0^\cdot \inner{(\beta^i - \beta^j)(Z_{t-})}{\ud Z^{\co, \qprob^j}_t}  \\
&+ \int_0^\cdot \pare{\int_{\Real^m} \frac{\gamma^i(Z_{t-}, y) - \gamma^j(Z_{t-}, y)}{\gamma^j(Z_{t-}, y)} \pare{\mu(\ud y, \ud t) - \nu_{\qprob^j} (Z_{t-}, \ud y) \ud t}}
, \quad \forall \iii.
\end{align*}
The latter implies that the processes $S^i$, when denominated in units of the asset $\jii$,  become local martingales on $\basisqj$ on each of the stochastic intervals $\dbra{0, \zeta_n}$ for $\nin$. It then follows in a straightforward way by use of Fatou's lemma that $\pare{S^i / S^j} \indic_{\set{S^j > 0}}$ is a nonnegative supermartingale on $\basisqj$ for all $\iii$. The behaviour of asset-price ratios in a general setting is taken up in Subsection \ref{subsec: ratios} below.

\subsection{Asset-price ratio processes}  \label{subsec: ratios}

Define the family of nonnegative processes
\begin{equation} \label{eq: rel}
R^{i j} \dfn \pare{\frac{S^i}{S^j} } \indic_{\set{S^j > 0}}, \quad  \iii \text{ and } \jii.
\end{equation}
In words, $R^{ij}$ represents the asset-price process $\iii$ denominated in units of the asset-price process $\jii$, as long as the latter asset has not defaulted yet.
By Theorem \ref{thm: prob_existnz}, for any $\iii$, $\jii$, and nonnegative optional process $H$ on $\basis$ and any $T \in \Stop$, it holds that
\begin{equation} \label{eq: measure_change_R}
S_0^j \expecqj \bra{R^{ij}_T H_T \on T < \zeta } = \expecp \bra{ S_T^i H_T \on S^j_T > 0, \, T < \zeta} = S_0^i \expecqi \bra{H_T \on S^j_T > 0, \, T < \zeta }.
\end{equation}
The next is a result in the spirit of the supermartingale optional sampling theorem.

\begin{prop} \label{prop: supermart_rel}
Under Assumption \ref{ass: main}, the process $R^{i j}$ is a (nonnegative)  supermatingale on $\basisqj$ for all $\iii$ and $\jii$.
\end{prop}

\begin{proof}
It suffices to show that $\expecqj \big[ R^{ij}_\tau  \big] \leq \expecqj \big[ R^{ij}_\sigma \big]$ holds for all fixed $\iii$, $\jii$ and $\sigma \in \Stop$ and $\tau \in \Stop$ with $\sigma \leq \tau$. Note that $R^{ij} = R^{i j} \indic_{\dbraco{0, \zeta}}$; therefore, we need to show that 
$\expecqj \big[ R^{ij}_\tau \on \tau < \zeta \big] \leq \expecqj \big[ R^{ij}_\sigma \on \sigma < \zeta \big]$. The first equality in \eqref{eq: measure_change_R} applied twice gives $S_0^j \expecqj \big[ R^{ij}_{\sigma} \on \sigma < \zeta  \big] = \expecp \big[ Y_\sigma S^i_{\sigma} \on S^j_{\sigma} > 0, \, \sigma < \zeta \big]$ and $S^j_0 \expecqj \big[ R^{ij}_{\tau} \on \tau < \zeta \big] = \expecp \big[ Y_\tau S^i_{\tau} \on S^j_{\tau} > 0, \, \tau < \zeta  \big]$. Therefore, $\expecqj \big[ R^{ij}_\tau \on \tau < \zeta \big] \leq \expecqj \big[ R^{ij}_\sigma \on \sigma < \zeta \big]$. is equivalent to $\expecp \big[ Y_\tau S^i_{\tau} \on S^j_{\tau} > 0, \, \tau < \zeta  \big] \leq \expecp \big[ Y_\sigma S^i_{\sigma} \on S^j_{\sigma} > 0, \, \sigma < \zeta  \big]$. Recall from Remark \ref{rem: no_expl_iff_mart} that, under Assumption \ref{ass: main}, $Y S^j$ is a nonnegative supermartingale on $\basisp$; therefore, it follows that $\prob \big[S^j_\sigma = 0, \, Y_\tau > 0, \, S^j_\tau > 0, \, \tau < \zeta \big] = 0$. The last fact combined with $\set{\tau < \zeta } \subseteq \set{\sigma < \zeta }$ implies the string of inequalities $Y_\tau \indic_{\set{S^j_\tau > 0, \, \tau < \zeta }} \leq Y_\tau \indic_{\set{S^j_\sigma > 0, \, \tau < \zeta}} \leq Y_\tau \indic_{\set{S^j_\sigma > 0, \, \sigma < \zeta }}$, holding modulo $\prob$. In turn, the last fact implies the first inequality in
\[
\expecp \big[ Y_\tau S^i_{\tau} \on S^j_{\tau} > 0, \, \tau < \zeta  \big] \leq \expecp \big[ Y_\tau S^i_{\tau} \on S^j_{\sigma} > 0, \, \sigma < \zeta   \big] \leq \expecp \big[ Y_\sigma S^i_{\sigma} \on S^j_{\sigma} > 0, \, \sigma < \zeta  \big],
\]
where the second equality follows from the fact that the process $Y S^i$ is a supermartingale on $\basisp$ and the optional sampling theorem for nonnegative supermartingales---see, for example, \cite[\S 1.3.C]{MR917065}. The proof is complete.
\end{proof}

In Section \ref{sec: valuation}, we shall make use of the family of random variables
\begin{equation} \label{eq: rho_ij}
\rho^{ij} \dfn \liminf_{t \uparrow \zeta} R^{ij}_{t}, \quad \iii \text{ and } \jii,
\end{equation}
where the notation ``$ \liminf_{t \uparrow \zeta}$'' is used to signify that a left-hand-side inferior limit is considered. If $(\zeta_n)_{\nin}$ is any sequence that foretells $\zeta$ under all $\qprob^i$, $\iii$, 
the nonnegative supermartingale convergence theorem \cite[\S 1.3.C]{MR917065} implies that, for all $\iii$ and $\jii$, on $\basisqj$ the $\F$-measurable random variable $\rho^{ij}$ is $\Real_+$-valued and the ``$\liminf$'' in \eqref{eq: rho_ij} is an actual limit.

\section{Valuation Formulas for Exchange Options} \label{sec: valuation}

\subsection{Valuation formulas for European-style exchange options}

Given the stochastic discount factor $Y$ of Assumption \ref{ass: main}, define the value of a European option to exchange asset $\iii$ for asset $\jii$ at time $T \in \Stop$ as
\begin{equation} \label{eq: eur_value_original}
\EX^{ij}(T) \dfn \expecp \bra{Y_T (S^j_T - S^i_T)_+  \on T < \zeta }.
\end{equation}
In view of Theorem \ref{thm: prob_existnz}, note the validity of the relationships $\EX^{ij}(T) \leq \expecp \big[ Y_T S^j_T \on T < \zeta \big] = S_0^j \qprob^j \big[ T < \zeta \big] \leq S_0^j$ for all $\iii$, $\jii$ and $T \in \Stop$.

\begin{rem}
Under Assumption \ref{ass: main}, $S^i_T = 0$ holds on $\set{\zeta \leq T}$ for all $\iii$. It follows that the indicator of the event $\set{T < \zeta}$ inside the expectation in \eqref{eq: eur_value_original} may be omitted. The same holds for several equations that will appear below (although not all); we choose to keep the indicator in order to explicitly reinforce the convention that no claims are honoured from time $\zeta$ onwards.
\end{rem}

The next result gives several representations for the value of European-style exchange options. Recall from \eqref{eq: rel} the definition of the collection of processes $R^{ij}$ for $\iii$ and $\jii$.

\begin{prop} \label{prop: Eur_option}
For all $\iii$, $\jii$ and $T \in \Stop$, the following formulas are valid:
\begin{eqnarray*}
\EX^{ij}(T)  &=& S^j_0 \qprob^j \big[ S^i_T < S^j_T , \ T < \zeta \big] - S^i_0 \qprob^i \big[ S^i_T < S^j_T, \ T < \zeta \big] \\
&=& S^j_0 \qprob^j \big[S^i_T \leq S^j_T, \ T < \zeta\big] - S^i_0 \qprob^i \big[S^i_T \leq S^j_T, \ T < \zeta\big] \\
&=& S^j_0  \expecqj \bra{ \big( 1 - R^{ij}_T \big)_+ \on T < \zeta } \\
&=& S^i_0  \expecqi \bra{ \big( R^{ji}_T - 1  \big)_+ \on T < \zeta } + S^j_0 \qprob^j \bra{S^i_T = 0, \, T < \zeta}.
\end{eqnarray*}
\end{prop}

\begin{proof}
Fix $\iii$ and $\jii$. Since $(S^j - S^i)_+ = S^j \indic_{\set{S^i < S^j}} - S^i \indic_{\set{S^i < S^j}} = S^j \indic_{\set{S^i \leq S^j}} - S^i \indic_{\set{S^i \leq S^j}}$, the first two equalities follow in a straightforward way from \eqref{eq: measure_change}. Continuing note that $(S^j - S^i)_+ = (S^j - S^i)_+ \indic_{\set{S^j > 0}} =  S^j (1 - R^{ij})_+$ holds. Using $H = (1 - R^{ij})_+$ in \eqref{eq: measure_change} (with $j$ replacing $i$ there), the third equality follows immediately. Furthermore, upon noting that $(S^j - S^i)_+ = (S^j - S^i)_+ \indic_{\set{S^i > 0}} + S^j \indic_{\set{S^i = 0}} = S^i (R^{ji} - 1)_+ + S^j \indic_{\set{S^i = 0}}$ and using \eqref{eq: measure_change} twice, once with $H = (R^{ji} - 1)_+$ and another time with $H = \indic_{\set{S^i = 0}}$ (and $j$ replacing $i$ there), the last equality follows.
\end{proof}

\begin{rem}
Fix $\jii$ and suppose that $\qprob^j \bra{\zeta < \infty} = 0$ holds, which in view of Remark \ref{rem: no_expl_iff_mart} is equivalent to the process $(Y_t S^j_t)_{\tir}$ being an actual martingale on $\basisp$. In that case, since $\qprob^j \bra{T < \zeta} = 1$ holds for all $T \in \Stop$ with $T < \infty$, a combination of 
Proposition \ref{prop: supermart_rel} and Proposition \ref{prop: Eur_option}, 
the convexity of the function $\Real \ni x \mapsto x_+ \in \Real_+$ and Jensen's inequality give $\EX^{ij} (\sigma) \leq \EX^{ij} (\tau)$ whenever $\sigma \in \Stop$ and $\tau \in \Stop$ are such that $\sigma \leq \tau < \infty$ holds. It follows that the value $\EX^{ij}(T)$ of the European exchange option is non-decreasing for finite maturities $T \in \Stop$.

In contrast to the situation where $\zeta$ is $\qprob^j$-a.s. infinite for some $\jii$, when $\qprob^j \bra{\zeta < \infty} > 0$ the previous monotonicity property need not hold, due to the non-triviality of the indicator of the event $\set{T < \zeta}$ in the expression $\EX^{ij}(T) = S^j_0  \expecqj \big[ \big( 1 - R^{ij}_T \big)_+ \on T < \zeta \big]$. The latter event is nonincreasing in $T$ and may result in reversal of the inequality $\EX^{ij} (\sigma) \leq \EX^{ij} (\tau)$ whenever $\sigma \in \Stop$ and $\tau \in \Stop$ are such that $\sigma \leq \tau < \infty$ holds. In fact, an example presented in \cite{MR2653260} shows a case where the function $\Real_+ \ni T \mapsto \EX^{ij} (T)$ is initially strictly increasing and then strictly decreasing. 
\end{rem}

\begin{rem}
The representation $\EX^{ij}(T) = S^j_0  \expecqj \big[ \big( 1 - R^{ij}_T \big)_+ \on T < \zeta \big]$ gives the value of the exchange option in terms of a put option on the asset $\iii$ by considering asset $\jii$ as a \num. Similarly, the expression $\EX^{ij}(T) = S^i_0  \expecqi \big[ \big( R^{ji}_T - 1  \big)_+ \on T < \zeta \big] + S^j_0 \qprob^j \bra{S^i_T = 0, \, T < \zeta}$ follows from the use of asset $\iii$ as a \num, in terms of a call option on asset $\jii$. 
Note however, an asymmetry between the two representations, since the equality $\EX^{ij}(T) = S^i_0  \expecqi \big[ \big( R^{ji}_T - 1  \big)_+ \on T < \zeta \big]$ is actually valid only if $\qprob^j \bra{S^i_T = 0, \, T < \zeta} = 0$ for $T \in \Stop$. 
\end{rem}

\subsection{Valuation formulas for American-style exchange options}

For $T \in \Stop$ define $\StopT$ as the class of all $\tau \in \Stop$ such that $0 \leq \tau \leq T$ holds. Given the process $Y$ of Assumption \ref{ass: main}, the value of an American option to exchange asset $\iii$ for asset $\jii$ up to time $T$ is defined to be\footnote{For a justification on why this definition of the value of an American-style option is reasonable, the interested reader can check \cite{MR2903625}.}
\begin{equation} \label{eq: amer_value_original}
\AX^{ij}(T) \dfn \sup_{\tau \in \StopT} \expecp \bra{Y_\tau (S^j_\tau - S^i_\tau)_+ \on \tau < \zeta  } = \sup_{\tau \in \StopT} \EX^{ij}(\tau).
\end{equation}
The inequalities $\EX^{ij} (T) \leq \AX^{ij} (T) \leq S^j_0$ hold for all $\iii$, $\jii$ and $T \in \Stop$. Proposition \ref{prop: Ame_option} provides, \emph{inter alia}, a formula for the early exercise premium $\AX^{ij} (T) - \EX^{ij} (T)$ of the American versus the European option. Recall from \eqref{eq: rho_ij} the random variables $\rho^{ij}$ for $\iii$ and $\jii$. 

\begin{prop} \label{prop: Ame_option}
Fix $\iii$, $\jii$, $T \in \Stop$, as well as any sequence $(\zeta_n)_{\nin}$ which foretells $\zeta$ under all $\pare{\qprob^i}_{\iii}$ (see Remark \ref{rem:foretellable_exists}). Then, the following are true:
\begin{enumerate}
	\item The sequence $\pare{\EX^{ij} (T \wedge \zeta_n)}_{\nin}$ is nondecreasing. Furthermore,
\begin{equation} \label{eq: Amer_approx}
\AX^{ij} (T) = \limn \EX^{ij} (T \wedge \zeta_n).
\end{equation}
	\item The early exercise premium is given by
\begin{equation} \label{eq: early_exe_premium}
\AX^{ij}(T) - \EX^{ij}(T) = S^j_0 \expecqj \big[ \big( 1 - \rho^{ij} \big)_+ \on \zeta \leq T \big].
\end{equation}
\end{enumerate}
\end{prop}

\begin{proof}
In the course of the proof, fix $\iii$, $\jii$ and $T \in \Stop$. 

\noindent \emph{(1).} Let $\tau \in \StopT$. By Proposition \ref{prop: Eur_option}, and since $\qprob^j \bra{\zeta_n < \zeta} = 1$ holds for all $\nin$, we obtain
\[
\EX^{ij} (\tau \wedge \zeta_n) = S^j_0 \expecqj \bra{\big(1 - R^{ij}_{\tau \wedge \zeta_n}\big)_+ \on \tau \wedge \zeta_n < \zeta} = S^j_0 \expecqj \bra{\big(1 - R^{ij}_{\tau \wedge \zeta_n}\big)_+}.
\]
The fact $\qprob^j \bra{\zeta_n < \zeta} = 1$ and Proposition \ref{prop: supermart_rel} imply the inequality $\expecqj \big[ R^{ij}_{\tau \wedge \zeta_m} \big] \leq \expecqj \big[ R^{ij}_{\tau \wedge \zeta_n} \big]$ whenever $\Natural \ni n \leq m \in \Natural$. The convexity of the function $\Real \ni x \mapsto x_+ \in \Real_+$ and Jensen's inequality imply that $\EX^{ij} (\tau \wedge \zeta_n) \leq \EX^{ij} (\tau \wedge \zeta_m)$ holds whenever $\Natural \ni n \leq m \in \Natural$, which shows that the sequence $\pare{\EX^{ij} (\tau \wedge \zeta_n)}_{\nin}$ is nondecreasing. Furthermore, in view of the fact that $\limn \zeta_n = \zeta$, it $\qprob^j$-a.s. holds that $(1 - R^{ij}_\tau)_+ \indic_{\set{\tau < \zeta }} \leq \liminf_{n \to \infty} \big( (1 - R^{ij}_{\tau \wedge \zeta_n})_+ \big)$. This fact,  coupled with Fatou's lemma, implies that
\[
\EX^{ij} (\tau) = \expecqj \bra{(1 - R^{ij}_\tau)_+ \on \tau < \zeta } \leq \expecqj \bra{\liminf_{n \to \infty} \pare{ (1 - R^{ij}_{\tau \wedge \zeta_n})_+ }} \leq \limn \EX^{ij} (\tau \wedge \zeta_n).
\]
In a similar way as was reasoned above, Proposition \ref{prop: supermart_rel} and the facts that $\qprob^j \bra{\zeta_n < \zeta} = 1$ for all $\nin$ and $\tau \leq T$ give $\EX^{ij} (\tau \wedge \zeta_n) \leq \EX^{ij} (T \wedge \zeta_n)$ for all $\nin$; therefore, $\EX^{ij} (\tau) \leq \limn \EX^{ij} (T \wedge \zeta_n)$ holds for all $\tau \in \StopT$. Equation \eqref{eq: Amer_approx} immediately follows.

\noindent \emph{(2).} Since $\limn R^{ij}_{T \wedge \zeta_n}  = \rho^{i j} \indic_{\set{\zeta \leq T}} + R^{ij}_T \indic_{\set{T < \zeta}}$ holds $\qprob^j$-a.s., the dominated convergence theorem gives
\[
\AX^{ij}(T) = \limn \EX^{ij} (T \wedge \zeta_n) = S_0^j \expecqj \bra{ \big( 1 - \rho^{ij} \big)_+ \on \zeta \leq T } + S_0^j \expecqj \bra{ \big( 1 - R^{ij}_T \big)_+ \on T < \zeta}.
\]
By Proposition \ref{prop: Eur_option}, the second term in the right-hand-side of the the above equation is equal to $\EX^{ij}(T)$; therefore, \eqref{eq: early_exe_premium} has been established.
\end{proof}

\begin{rem} \label{rem: Ame_approx_same}
Proposition \ref{prop: Ame_option} implies that, for any $T \in \Stop$, the supremum in \eqref{eq: amer_value_original} for $\AX^{ij}(T)$ is monotonically achieved through the sequence $\pare{T \wedge \zeta_n}_{\nin}$ of stopping times in $\StopT$, this being true for all combinations of $\iii$ and $\jii$. This fact has the important consequence that a parity for American exchange options follows from the corresponding parity for European options---see the statement and proof of Proposition \ref{prop: parity_main}.
\end{rem}

\begin{rem}
While in the Black-Scholes-Merton modelling environment discussed in \cite{RePEc:bla:jfinan:v:33:y:1978:i:1:p:177-86} it is never optimal to exercise an American-style exchange option before a \emph{finite} maturity $T \in \Stop$, Proposition \ref{prop: Ame_option} implies that, if $\qprob^j \bra{\zeta \leq T} > 0$ holds, it is not optimal to keep an American option to exchange any asset $\iii$ for some asset $\jii$ until maturity $T \in \Stop$. Instead, \eqref{eq: Amer_approx} reasonably suggests that one should keep the option until maturity $T \in \Stop$ \emph{provided} that the end of the whole economy does not appear imminent; otherwise, early exercise may be preferable.
\end{rem}

\begin{rem}
Formulas like \eqref{eq: Amer_approx} have appeared in \cite{MR2276895}, as ``corrected'' values for European-style options. In fact, Proposition \ref{prop: Ame_option} implies that they correspond to values of American-style options.
\end{rem}

\begin{rem}
Using the (self-explanatory) notation $R^{ij}_{T \wedge  (\zeta -) } = R^{ij}_T \indic_{\set{T < \zeta}} + \rho^{ij} \indic_{\set{\zeta \leq T}}$ for $\iii$, $\jii$ and $T \in \Stop$, it follows by a combination of Proposition \ref{prop: Eur_option} and Proposition \ref{prop: Ame_option} that
\[
\AX^{ij} (T) = S_0^j \expecqj \bra{ \pare{1 - R^{ij}_{T \wedge  (\zeta -) }}_+ },
\]
which provides a direct representation for the value of American-style exchange options.
\end{rem}

\begin{rem}
The formulas in Propositions \ref{prop: Eur_option} and \ref{prop: Ame_option} open the way in the numerical approximation of European and American exchange option values, as well as early exercise premia. Indeed, in the setting of Example \ref{exa: Markov_factor_reprise} (which continues the discussion in Subsection \ref{exa: Markov_factor}) one can use standard Monte-Carlo simulation techniques in order to identify the corresponding expectations; one simply needs to identify $\zeta$ with $\zeta_n$ for some large $\nin$, for the sequence $(\zeta_n)_{\nin}$ which is given in \eqref{eq: zeta_n_explicit}. This procedure can also be used for calibration of parametric models to match European and American exchange option prices observed in the market.
\end{rem}

An interesting special case in Proposition \ref{prop: Ame_option} is when $\qprob^j \bra{\rho^{ij} = 0} = 1$ holds for some $\iii$ and $\jii$; this is, for example, true in the case in the Black-Scholes-Merton model where the logarithms of asset-price processes are (not perfectly) correlated drifted Brownian motions. When $\qprob^j \bra{\rho^{ij} = 0} = 1$ holds for $\iii$ and $\jii$, the simpler formula $\AX^{ij}(T) - \EX^{ij}(T) = S^j_0 \qprob^j[\zeta \leq T]$ for the early exercise premium holds for all $T \in \Stop$. The next result gives several equivalent formulations of the latter condition.

\begin{prop} \label{prop: ass_amer}
Fix $\iii$ and $\jii$, as well as any sequence $(\zeta_n)_{\nin}$ which foretells $\zeta$ under all $\pare{\qprob^i}_{\iii}$ (see Remark \ref{rem:foretellable_exists}). Under Assumption \ref{ass: main}, the following statements are equivalent:
\begin{enumerate}
	\item $\limn \AX^{ij}(\zeta_n) = S^j_0$.
	\item $\limn \EX^{ij}(\zeta_n) = S^j_0$.
	\item $\limn \qprob^j \big[ S^j_{\zeta_n} \leq S^i_{\zeta_n} \big] = 0$ and $\limn \qprob^i \big[ S^i_{\zeta_n} \leq S^j_{\zeta_n} \big] = 0$.
	\item $\qprob^j \bra{ \rho^{ij}  = 0} = 1$.
	\item $\AX^{ij}(T) - \EX^{ij}(T) = S^j_0 \qprob^j[\zeta \leq T]$ holds for all $T \in \Stop$.
\end{enumerate}
\end{prop}

\begin{proof}
Fix $\iii$ and $\jii$. By Proposition \ref{prop: Ame_option}, $\AX^{ij} (\zeta_n) = \lim_{m \to \infty} \EX^{ij} (\zeta_n \wedge \zeta_m) = \EX^{ij} (\zeta_n)$ holds for all $\nin$. This shows the equivalence of statements (1) and (2). Furthermore, since $\qprob^i \bra{\zeta_n < \zeta} = 1$ and $\qprob^j \bra{\zeta_n < \zeta} = 1$ holds for all $\nin$, $\EX^{ij} (\zeta_n) = S^j_0 \qprob^j \big[ S^i_{\zeta_n} < S^j_{\zeta_n}  \big] - S^i_0 \qprob^i \big[ S^i_{\zeta_n} < S^j_{\zeta_n} \big]$ follows from Proposition \ref{prop: Eur_option}. Therefore,  $\limn \EX^{ij}(\zeta_n) = S^j_0$ is equivalent to the validity of both $\limn \qprob^j \big[ S^j_{\zeta_n} \leq S^i_{\zeta_n} \big] = 0$ and $\limn \qprob^i \big[ S^i_{\zeta_n} < S^j_{\zeta_n} \big] = 0$. Since
\[
S^j_0 \qprob^j \big[ S^j_{\zeta_n}  = S^i_{\zeta_n} \big] = \prob \big[ S^j_{\zeta_n} = S^i_{\zeta_n}, \, \zeta_n < \zeta  \big] = S^i_0 \qprob^j \big[ S^j_{\zeta_n} = S^i_{\zeta_n} \big]
\]
holds in view of Theorem \ref{thm: prob_existnz}, $\limn \EX^{ij}(\zeta_n) = S^j_0$ is equivalent to $\limn \qprob^j \big[ S^j_{\zeta_n} \leq S^i_{\zeta_n} \big] = 0$ and $\limn \qprob^i \big[ S^i_{\zeta_n} \leq S^j_{\zeta_n} \big] = 0$. This shows the equivalence of (2) and (3). Therefore, the equivalence of conditions (1), (2) and (3) has been established. Continuing, a combination of Proposition \ref{prop: Eur_option} and the dominated convergence theorem give $\limn \EX^{ij} (\zeta_n) = S_0^j \expecqj \bra{(1 - \rho^{ij})_+}$. Therefore, conditions (2) and (4) are equivalent. The fact that condition (4) implies condition (5) follows from \eqref{eq: early_exe_premium}. Furthermore, if (5) holds then \eqref{eq: early_exe_premium} with $T = \zeta$ gives $\expecqj \big[ \pare{1 - \rho^{ij}}_+ \big] = 1$, which is equivalent to $\qprob^j \bra{\rho^{ij} = 0} = 1$, i.e., condition (4).
\end{proof}

\begin{rem}
Note that condition (3) of Proposition \ref{prop: ass_amer} is symmetric in $\iii$ and $\jii$. This means that conditions (1), (2), (4) and (5) of Proposition \ref{prop: ass_amer} are also equivalent to the corresponding conditions where the roles of $i$ and $j$ are interchanged.
\end{rem}

\begin{rem}
Fix $\iii$ and $\jii$. Under any of the equivalent conditions of Proposition \ref{prop: ass_amer}, the equality $\AX^{ij}(T) = S^j_0$ holds whenever $T \in \Stop$ is such that $T \geq \zeta$. In fact, one can get an expression for the difference $S^j_0  - \AX^{ij}(T)$ for all $T \in \Stop$. Assuming any of the equivalent conditions of Proposition \ref{prop: ass_amer}, $S_0^j - \AX^{ij}(T) = S_0^j \qprob^j \bra{T < \zeta} - \EX^{ij}(T)$ holds for all $T \in \Stop$. Since $\EX^{ij}(T) = S_0^j \expecqj \big[ (1 - R^{ij}_T)_+ \on T < \zeta \big]$ holds by Proposition \ref{prop: Eur_option}, we obtain
\[
S_0^j - \AX^{ij}(T) = S_0^j \expecqj \bra{ 1 \wedge R^{ij}_T  \on T < \zeta} = S_0^j \qprob^j \bra{ R^{ij}_T \geq 1, \, T < \zeta} + S_0^j \expecqj \bra{ R^{ij}_T \on R^{ij}_T < 1, \, T < \zeta}.
\]
Now, $\qprob^j \big[ R^{ij}_T \geq 1, \, T < \zeta \big] = \qprob^j \big[ S^j_T \leq S^i_T, \, S^j_T > 0, \, T < \zeta \big] = \qprob^j \big[ S^j_T \leq S^i_T, \ T < \zeta \big]$, the last equality following from $\qprob^j \big[ S^j_T =0, T < \zeta \big] = 0$ in Remark \ref{rem: zero_prob_zero}. Furthermore, note that \eqref{eq: measure_change_R} gives
\[
S_0^j \expecqj \bra{ R^{ij}_T \on R^{ij}_T < 1, \, T < \zeta} = S_0^i \qprob^i \big[ S^i_T < S^j_T, \, S^j_T > 0, \, T < \zeta \big] = S_0^i \qprob^i \big[ S^i_T < S^j_T,  \, T < \zeta \big],
\]
the last equality following from the nonnegativity of $S^i$. It follows that
\[
S_0^j - \AX^{ij}(T) = S^j_0 \qprob^j \big[ S^j_T \leq S^i_T, \, T < \zeta \big] + S^i_0 \qprob^i \big[ S^i_T < S^j_T, \, T < \zeta \big].
\]
\end{rem}

\section{Parities Involving Exchange Options} \label{sec: parity}

\subsection{Parities} The following result gives two parities---one regarding European-style and another regarding American-style exchange options.

\begin{prop} \label{prop: parity_main}
Let $\iii$ and $\jii$, as well as $T \in \Stop$. Under Assumption \ref{ass: main}, the following parities hold:
\begin{eqnarray}
\EX^{ij} (T) + S_0^i \qprob^i \bra{T < \zeta} &=& \EX^{ji} (T) + S_0^j \qprob^j \bra{T < \zeta}, \label{eq: parity_eur_eur} \\
\AX^{ij} (T) + S_0^i &=& \AX^{ji} (T) + S_0^j \label{eq: parity_ame_ame}.
\end{eqnarray}
\end{prop}

\begin{proof}
Combining the relationships $\EX^{ij}(T) = S^j_0 \qprob^j \big[ S^i_T < S^j_T, \, T < \zeta \big] - S^i_0 \qprob^i \big[ S^i_T < S^j_T, \, T < \zeta \big]$ and $\EX^{ji}(T) = S^i_0 \qprob^i \big[ S^j_T \leq S^i_T, \, T < \zeta \big] - S^j_0 \qprob^j \big[ S^j_T \leq S^i_T, \, T < \zeta \big]$, both following from Proposition \ref{prop: Eur_option}, one obtains $\EX^{ij}(T) - \EX^{ji}(T) = S_0^j \qprob^j \bra{T < \zeta} - S_0^i \qprob^i \bra{T < \zeta}$, which shows \eqref{eq: parity_eur_eur}. Let $(\zeta_n)_{\nin}$ be a sequence which foretells $\zeta$ under all $\pare{\qprob^i}_{\iii}$. Replacing $T$ by $T \wedge \zeta_n$ and using the fact that $\qprob^i \bra{\zeta_n < \zeta} = 1 = \qprob^j \bra{\zeta_n < \zeta}$ holds for all $\nin$, we obtain $\EX^{ij} (T \wedge \zeta_n) + S_0^i = \EX^{ji} (T \wedge \zeta_n) + S_0^j$. Sending $n$ to infinity and using \eqref{eq: Amer_approx}, \eqref{eq: parity_ame_ame} follows. 
\end{proof}

\begin{rem}
An alternative, more direct proof of \eqref{eq: parity_eur_eur} utilizes the equality
\begin{equation} \label{eq: parity_pro}
(S^j - S^i)_+ + S^i = (S^i - S^j)_+ + S^j, \quad \text{for } \iii \text{ and } \jii.
\end{equation}
Applying \eqref{eq: parity_pro} with the processes sampled at $T \in \Stop$ on the event $\set{T < \zeta}$, multiplying both sides by $Y_T$ and taking expectation with respect to $\prob$, one obtains \eqref{eq: parity_eur_eur} by Proposition \ref{prop: Eur_option}, given the equalities $\expecp \big[ Y_T S^i_T \on T < \zeta \big] = S_0^i \qprob^i \bra{T < \zeta}$ and $\expecp \big[ Y_T S^j_T \on T < \zeta \big] = S_0^j \qprob^j \bra{T < \zeta}$ that follow from \eqref{eq: measure_change}.

For $\iii$, the quantity $S_0^i \qprob^i \bra{T < \zeta}$ is the value of the contract that pays $S^i_T$ at time $T \in \Stop$ when $T < \zeta$. Being a European-style contract, its value may be \emph{strictly} less than $S_0^i$, which happens exactly when $\qprob^i \bra{\zeta \leq T} > 0$. In contrast, the value of the corresponding ``American'' option that pays $S^i_\tau$ at any chosen time $\tau \in \StopT$ for $T \in \Stop$ would be
\begin{equation} \label{eq: Ame_stock}
\sup_{\tau \in \StopT} \expecp \bra{Y_\tau S^i_\tau \on \tau < \zeta} = \sup_{\tau \in \StopT}  S_0^i \qprob^i \bra{\tau < \zeta} = S_0^i \qprob^i \bra{\zeta > 0} = S_0^i,
\end{equation}
since $\zeta > 0$ holds identically (recall the set-up of Subsection \ref{subsec: prelims}.). In models where $\qprob^i \bra{\zeta < \infty} = 0$ is valid for all $\iii$, $\EX^{ij} (T) = \AX^{ij}(T)$ holds for all $\iii$, $\jii$ and $T \in \Stop$ with $T < \infty$. Then, \eqref{eq: parity_ame_ame} becomes a parity for both American-style and European-style exchange options (the latter upon replacing $\AX^{ij}$ by $\EX^{ij}$ and $\AX^{ji}$ by $\EX^{ji}$). The fact that \eqref{eq: parity_eur_eur}, instead of \eqref{eq: parity_ame_ame}, holds for European options has sometimes lead to claims that the ``usual'' parity is not valid in markets where bubbles exist. Of course, in order for a parity to hold, the contracts used have to be of similar type. In this sense, \eqref{eq: parity_eur_eur} is the correct and perfectly valid parity for European options; this has already been made clear in \cite{MR2732840}, in the setting of the example of Subsection \ref{sec: example} below. On the other hand, when American-style exchange options are involved, American-style contracts that pay off the stock price have to be used in both sides; in view of \eqref{eq: Ame_stock}, \eqref{eq: parity_ame_ame} is the parity to be expected. As noted in Remark \ref{rem: Ame_approx_same} and demonstrated in the proof of Proposition \ref{prop: parity_main}, the American parity \eqref{eq: parity_ame_ame} follows from the validity of \eqref{eq: parity_eur_eur} and the fact that the approximating sequence $(T \wedge \zeta_n)_{\nin}$ is the same for all choices of $\iii$ and $\jii$.
\end{rem}

In the special case where any of the equivalent conditions of Proposition \ref{prop: ass_amer} hold, two more parities are valid, mixing European and American options.

\begin{prop} \label{prop: more put_call}
Under Assumption \ref{ass: main} and the validity of any of the equivalent conditions of Proposition \ref{prop: ass_amer}, the following parities hold:
\begin{eqnarray*}
\AX^{ij} (T) + S_0^i \qprob^i \bra{T < \zeta} &=& \EX^{ji} (T) + S_0^j, \\
\EX^{ij} (T) + S_0^i  &=& \AX^{ji} (T) + S_0^j \qprob^j \bra{T < \zeta}.
\end{eqnarray*}
\end{prop}

\begin{proof}
Since Proposition \ref{prop: ass_amer} gives $\AX^{ij} (T) = \EX^{ij} (T) + S_0^j \qprob^j \bra{T < \zeta}$ and $\AX^{ji} (T) = \EX^{ji} (T) + S_0^i \qprob^i \bra{T < \zeta}$, both relationships follow directly from \eqref{eq: parity_eur_eur}.  
\end{proof}

\begin{rem}
The underlying reason for the parities in Proposition \ref{prop: more put_call} under any of the equivalent conditions of Proposition \ref{prop: ass_amer} is that the early exercise premium of the exchange option with payoff $(S^j_T - S^i_T)_+$ at time $T \in \Stop$ for $\iii$ and $\jii$ coincides with the difference between the asset price $S^j_0$ and $S^j_0 \qprob \bra{\zeta > T}$, the latter being the ``European value of a claim that pays $S^j_T$ at time $T$.''
\end{rem}

\subsection{An illustrative example involving the three-dimensional Bessel process} \label{sec: example}
Consider the case where $E = (0, \infty)$ and $\prob$ is such that $Z$ under $\prob$ is behaving like a three-dimensional Bessel process with unit initial value. Note that $\prob \bra{\zeta < \infty} = 0$. Let $I = \set{0,1}$, and suppose that $S^0 = K \indic_{\dbraco{0, \zeta}}$ for some $K \in (0, \infty)$ and $S^1 \equiv Z \indic_{\dbraco{0, \zeta}}$. It can be shown in a straightforward way that $Y = (1 / Z) \indic_{\dbraco{0, \zeta}}$ is the (essentially, modulo $\prob$-evanescence) unique process such that $Y S^i$ is a local martingale on $\basisp$ for $i \in I$. Clearly $\qprob^1 = \prob$, while $\qprob^0$ can be seen to coincide with the probability on $\F$ such that $Z$ is Brownian motion starting from one and killed when it reaches zero. The equality $\qprob^1 \bra{\rho^{01} = 0} = \prob \bra{\lim_{t \to \infty} Z_t = \infty} = 1$ follows from the fact that $Z$ behaves like three-dimensional Bessel process under $\prob$. In particular, we obtain all relations of Proposition \ref{prop: ass_amer} when $i = 0$ and $j =1$, as well as when $i = 1$ and $j = 0$.

As $\qprob^1 \bra{\zeta < \infty} = \prob \bra{\zeta < \infty} = 0$, it follows that $\AX^{01} (T) = \EX^{01} (T)$ holds for all $T \in \Real_+$. Furthermore, Proposition \ref{prop: Eur_option} gives $\EX^{01} (T) = \prob \bra{Z_T > K} - K \qprob^0 \bra{Z_T > K, \zeta > T}$ for $T \in \Real_+$. Let $\Phi : \Real \mapsto (0,1)$ denote the cumulative distribution function of the standard normal law, and set $\oPhi = 1 - \Phi$. The joint distribution of Brownian motion and its minimum gives
\[
\qprob^0 \bra{Z_T > K, \zeta > T} = \Phi \pare{ \frac{1 - K}{\sqrt{T}} } - \oPhi \pare{ \frac{1 + K}{\sqrt{T}} }, \quad T \in \Real_+.
\]
Furthermore, from properties of the non-central chi-squared distribution one can obtain that
\[
\prob \bra{Z_T > K} = \Phi \pare{ \frac{1 - K}{\sqrt{T}} } + \oPhi \pare{ \frac{1 + K}{\sqrt{T}} } + \sqrt{\frac{2 T}{\pi}} \exp \pare{- \frac{1 + K^2}{2 T}} \sinh \pare{ \frac{K}{T}}, \quad T \in \Real_+.
\]
(For the last formula see also \cite[Proposition 1]{MR2732840}.) It then follows that
\[
\EX^{01} (T) = (1 + K) \oPhi \pare{ \frac{1 + K}{\sqrt{T}} } + (1- K) \Phi \pare{ \frac{1 - K}{\sqrt{T}} } + \sqrt{\frac{2 T}{\pi}} \exp \pare{- \frac{1 + K^2}{2 T}} \sinh \pare{ \frac{K}{T}}, \quad T \in \Real_+,
\]
with the same equality valid for $\AX^{01} (T)$. Equation \eqref{eq: parity_ame_ame} gives $\AX^{10} (T) = \AX^{01} (T) - (1 - K)$, i.e.,
\[
\AX^{10} (T) = (1 + K) \oPhi \pare{ \frac{1 + K}{\sqrt{T}} } - (1- K) \oPhi \pare{ \frac{1 - K}{\sqrt{T}} } + \sqrt{\frac{2 T}{\pi}} \exp \pare{- \frac{1 + K^2}{2 T}} \sinh \pare{ \frac{K}{T}}, \quad T \in \Real_+.
\]
Furthermore, the law of the minimum of Brownian motion gives $\qprob^0 \bra{\zeta \leq T} = 2 \oPhi(1 / \sqrt{T})$ holding $T \in \Real_+$, which implies that $\EX^{10} (T) = \AX^{10} (T) - 2 K \oPhi(1 / \sqrt{T})$ holds for $T \in \Real_+$.

Note that the previous closed-form expressions give $\lim_{T \to \infty} \EX^{01} (T) = 1 = \lim_{T \to \infty} \AX^{01} (T)$, as well as $\lim_{T \to \infty} \EX^{10} (T) = 0 < K = \lim_{T \to \infty} \AX^{10} (T)$.

\bibliographystyle{alpha}
\bibliography{exchange}
\end{document}